\documentclass[11pt,english]{article}
 \usepackage[margin=0.75in]{geometry}
 \usepackage{hyperref}
 \hypersetup{
     colorlinks=true,
     linkcolor=blue,
     citecolor=blue,
 }
\usepackage[ruled,vlined,linesnumbered]{algorithm2e}

\usepackage{pst-node}
\usepackage{pst-coil}
\usepackage{auto-pst-pdf}
 \usepackage[english]{babel}
\usepackage{bm,xfrac,amsmath,amsthm,soul, comment, xspace, mathtools, amsfonts}
\usepackage{footnote}
\usepackage{nicefrac,bm}
\usepackage{thmtools}
\usepackage{thm-restate}
\usepackage{graphicx}
\usepackage{pifont,multirow}
\usepackage{textcomp}
\usepackage[switch]{lineno}

\title{Approximating APS Under Submodular and XOS Valuations\\ with Binary Marginals}

\author{
Pooja Kulkarni\thanks{University of Illinois at Urbana-Champaign.}\\ \texttt{poojark2@illinois.edu} \and
Rucha Kulkarni\thanks{University of Illinois at Urbana-Champaign.}\\\texttt{ruchark2@illinois.edu} \and
Ruta Mehta\thanks{University of Illinois at Urbana-Champaign}\\ \texttt{rutameht@illinois.edu}
}

\let\origappendix\appendix 
\renewcommand\appendix{\clearpage\pagenumbering{arabic}\origappendix}
\theoremstyle{plain}
\newtheorem{remark}{Remark}[section]

\newtheorem{lemma}{Lemma}[section]
\newtheorem*{lemma*}{Lemma}

\newtheorem{theorem}{Theorem}[section]

\newtheorem{claim}{Claim}[section]
\newtheorem*{claim*}{Claim}

\newtheorem{definition}{Definition}[section]

\newcommand{\A}{\mathcal{A}}
\newcommand{\T}{\mathcal{T}}
\newcommand{\Set}{\mathcal{S}}

\newcommand{\M}{{\mathcal M}}
\newcommand{\Vals}{(v_i)_{i\in [n]}}
\newcommand{\I}{{\mathcal I}}
\newcommand{\F}{\mathcal{F}}

\newcommand{\Ins}{([n],[m],\Vals)}

\newcommand{\vecp}{\mathbf{p}}

\newcommand{\classNP}{{\sf NP}}

\newcommand{\threematch}{{{\sf 3\text{-}D\text{-}MATCHING}}\xspace}

\newcommand{\EF}{{\sf{EF1}}\xspace}

\newcommand{\Prop}{{\sf{Prop1}}\xspace}
\newcommand{\Propx}{{\sf{Propx}}\xspace}
\newcommand{\MMS}{{\sf{MMS}}\xspace}
\newcommand{\APS}{{\sf{APS}}\xspace}
\newcommand{\MSW}{{\sf{MSW}}\xspace}

\newcommand{\PA}{{{A^\pi}}}

\newcommand{\R}{\mathcal R}

\newcommand{\XOS}{{\sf{XOS}}\xspace}
\newcommand{\MRF}{{\sf{MRF}}\xspace}
\newcommand{\MRFs}{{\sf{MRFs}}\xspace}

\newif\ifsoda
\sodatrue

\let\oldnl\nl
\newcommand{\nonl}{\renewcommand{\nl}{\let\nl\oldnl}}

\DeclareMathOperator*{\argmax}{\arg\!\max}

\renewcommand{\nonl}{\renewcommand{\nl}{\let\nl\oldnl}}
\long\def\symbolfootnote[#1]#2{\begingroup%
\def\thefootnote{\fnsymbol{footnote}}\footnote[#1]{#2}\endgroup}

\newcommand{\Price}{\mathcal{P}}

\begin{document}
\maketitle
\begin{abstract}

We study the problem of fairly dividing indivisible goods among a set of agents under the fairness notion of Any Price Share (\APS). \APS is known to dominate the widely studied Maximin share (\MMS). Since an exact \APS allocation may not exist, the focus has traditionally been on the computation of approximate \APS allocations. \cite{BabaioffEF21} studied the problem under additive valuations, and asked $(i)$ how large can the $\APS$ value be compared to the $\MMS$ value? and $(ii)$ what guarantees can one achieve beyond additive functions. We partly answer these questions by considering valuations beyond additive, namely submodular and \XOS functions, with binary marginals. 

For the submodular functions with binary marginals, also known as matroid rank functions (\MRFs), we show that \APS is exactly equal to \MMS. Consequently, following \cite{BarmanBKS20} we show that an exact $\APS$ allocation exists and can be computed efficiently while maximizing the social welfare. Complementing this result, we show that it is $\classNP$-hard to compute the \APS value within a factor of $5/6$ for submodular valuations with three distinct marginals of $\{0, \frac{1}{2}, 1\}$.

We then consider binary \XOS functions, which are immediate generalizations of binary submodular functions in the complement free hierarchy. In contrast to the \MRFs setting, \MMS and \APS values are not equal under this case. Nevertheless, we can show that they are only a constant factor apart. In particular, we show that under binary \XOS valuations, $\MMS \leq \APS \leq 2 \cdot \MMS + 1$. Further, we show that this is almost the tightest bound we can get using \MMS, by giving an instance where $\APS \geq 2 \cdot \MMS$. The upper bound on \APS, combined with \cite{li2021fair}, implies a $~0.1222$-approximation for $\APS$ under binary $\XOS$ valuations. And the lower bound implies the non-existence of better than $0.5$-\APS even when agents have identical valuations, which is in sharp contrast to the guaranteed existence of exact \MMS allocation when agent valuations are identical. 
\end{abstract}
\paragraph{Acknowledgements.} Pooja Kulkarni, Rucha Kulkarni and Ruta Mehta are supported by NSF CAREER Award CCF 1750436.
\section{Introduction}

Finding fair allocations of indivisible resources is a central problem within economics, game theory, social choice theory, and computer science. Given a set $[n]$ of agents and a set $[m]$ of indivisible items, the problem asks to partition the items among the agents in a \emph{fair} manner. Preferences of each agent $i\in[n]$ for bundles of goods are represented by monotone valuation functions $v_i:2^{[m]} \rightarrow \mathbb R_+$. A formal study of fair division began with the work of \cite{steinhaus1948problem}. Since then, several notions of fairness like \MMS, \EF and \Prop have been introduced and are well-studied (See ~\cite{AmanatidisBFV22survey, AzizLMW22survey} for surveys on these). 

Any Price Share (\APS) is one of the more recently introduced \cite{BabaioffEF21} notions, and has already garnered significant interest (See Section \ref{sec:rel-work} for a brief review). The \APS value of an agent is defined as the maximum value she can obtain with a budget of $1/n,$ given any vector of prices of goods that sums to $1.$ An allocation where every agent gets at least her \APS value is called an {\em \APS allocation}. An attractive feature of \APS is that it is independent of the valuations of the other agents. This is the same with the well-studied notion of maximin share (\MMS)~\cite{ghodsi2018fair, garg2018approximating, BarmanV21, li2021fair}. The \APS value of any agent is known to dominate their \MMS value~\cite{BabaioffEF21}. As \MMS allocations are known to not exist~\cite{ProcacciaW14} even in the additive valuations case\footnote{the value of a set equals the sum of values of goods in the set.}, the same holds for $\APS$ allocations. We, therefore, focus on the problem of finding approximate \APS, or $\alpha$-\APS allocations, which give every agent a bundle of value at least $\alpha$ times their \APS, for some $\alpha>0.$

The problem of finding approximate \APS allocations has been studied for the case of additive valuations~\cite{BabaioffEF21} and more recently for submodular valuations \cite{ben2023fair}. Numerous real world applications of fair division, like public housing to ethnic minorities, assigning kindergarten slots, course seat assignments for classes with capacity constraints and where students can specify preferences for a fixed maximum number of classes, require the valuation functions to be beyond additive and capture a \textit{diminishing marginal returns} property: this essentially means that the marginal value of a good over a set of goods diminishes over supersets of the set. Submodular functions capture this very natural property and are therefore considered a fundamental class of valuations. Fractionally subadditive (\XOS) functions are immediate generalizations of submodular functions in the complement-free hierarchy. We study $\APS$ under both these function classes, with the constraint that the marginal values are \emph{binary}, meaning the marginal value of any good over any subset of goods is either zero or one. These classes have a rich structure and have been well-studied for other fairness notions and also in optimization theory (See Section \ref{sec:rel-work} for a brief review). We partly resolve the following questions posed in \cite{BabaioffEF21} for submodular and \XOS functions with binary marginal values. 


\textit{Question 1: How far apart can the \MMS and \APS values of an agent be in any instance?}

\textit{Question 2: What guarantees for \APS can we ensure beyond additive valuations?}

Submodular functions with binary marginal values are equivalently known as Matroid rank functions (\MRFs), and are widely studied, for instance~\cite{schrijver2003combinatorial,Shioura12, benabbou2021finding, BarmanV21}. We show the following surprising result. Although the \APS value is known to be strictly higher than the \MMS value even for general additive functions \cite{BabaioffEF21}, we show that for any \MRF, these two notions are equivalent. This immediately leads to polynomial time algorithms to obtain exact \APS allocations using the algorithms that obtain exact \MMS allocations \cite{BarmanV21}. In fact, the known algorithms also ensure economic efficiency, by giving \APS allocations that simultaneously maximize social welfare (the total sum of values received by agents).

We then analyze the classic generalization of binary submodular functions in the complement free hierarchy, namely fractionally subadditive functions (\XOS) with binary marginals. For this, we show that the \APS value of an agent with such a valuation function is at most $2\mu+1,$ where $\mu$ is her \MMS value. Using this fact together with $0.3666$-\MMS allocation computing algorithm by \cite{li2018fair} we obtain an efficient algorithm for computing $0.1222$-\APS allocation. In contrast to the relation between \MMS and \APS for \MRFs, we show that there exist instances with only two agents and identical valuations in this setting such that $\APS \ge 2 \MMS$. As a consequence we get that even under identical valuations with two agents, better than $0.5$-\APS allocation may not exist. This is in sharp contrast to $\MMS$ where by definition an exact $\MMS$ allocation exists when agents have identical valuations. 

Finally, we show that if binary submodular functions are generalized to allow three distinct marginal values, in $\{0,1/2,1\},$ instead of the two values $\{0,1\}$ in \MRFs, then the problem of computing $\alpha$-\APS allocations, even among agents with identical valuation functions, for any factor $\alpha$ better than $5/6,$ is \classNP-hard. Equivalently, this means that the problem of computing the \APS values approximately up to a factor better than $5/6$ for such settings, is \classNP-hard.

Our results can be summarized as follows.
\begin{itemize}
    \item \APS = \MMS for submodular functions with binary marginals. Exact \APS values and allocations that give \APS along with maximum social welfare can thus be efficiently computed.
    \item \APS$\leq2\cdot$\MMS$+1,$ for $\XOS$ functions with binary marginals. A $0.1222$-approximate \APS allocation can thus be efficiently computed.
    \item There exist instances with identical binary, \XOS valuations where $0.5$-\APS allocation does not exist.
    \item Submodular functions with ternary marginals: computing \APS values approximately to a factor better than $5/6$ is \classNP-hard.
\end{itemize}

\subsection{Further Related Work}\label{sec:rel-work}
\paragraph{\APS}\APS was introduced in \cite{BabaioffEF21}, who also prove that \APS dominates \MMS for non-negative valuations. Further, they give $0.667$-\APS allocation for goods, and a $2$-\APS allocation for chores, under additive valuations. \cite{LiLW21} study the connection of the known notion of \Propx with \APS for chores, giving a $2$-\APS allocation here. \cite{ChakrabortySS22} compare \APS with other \textit{share based} notions for the case of agents with asymmetric entitlements. \cite{FeigeT22} study group fairness guarantees with \APS, under additive valuations. \cite{ben2023fair} recently gave a $\frac{1}{3}$-approximate algorithm for computing $\APS$ with submodular valuations for asymmetric agents i.e., when agents have different entitlements.

\paragraph{Matroid Rank Functions: } Rank functions of matroid are one of the fundamental set functions and the optimization of these functions has been studied in detail, see \cite{schrijver2003combinatorial}. \cite{benabbou2021finding} and \cite{barman2020existence} identify multiple domains where matroid rank functions show up naturally like fair allocation of public housing units. These functions have been studied in context of fair division, for other fairness and efficiency like Nash Social Welfare\footnote{Nash welfare is the geometric mean of agent's valuations} \cite{BabaioffEF21}  {\sf EF1}
\footnote{An agent values her bundle more than other agent's bundle up to removal of one (some) good}\cite{BabaioffEF21, benabbou2021finding}, \MMS \cite{barman2020existence} and combinations of these~\cite{ViswanathanZ22}. Notably, polynomial time algorithms that output the optimal allocations under all these fairness notions are known in the respective works.
\paragraph{Binary \XOS valuations}
Binary \XOS valuations generalize matroid rank functions and have been studied in fair division context. \cite{li2018fair} give an algorithm that gives a 0.3667-\MMS. \cite{barman2021approximating} give a 288-approximation algorithm for maximizing Nash social welfare under these valuations.

\section{Notation and Preliminaries}\label{sec:prelim}
\textit{Notation. $[k]$ denotes the set $\{1,2,\cdots,k-1,k\}.$}

\noindent
\textbf{Model.} We study the problem of fairly dividing a set of $m$ indivisible goods, among $n$ agents. 
Preferences of an agent $i \in [n]$ is defined by a valuation function 
$v_i: 2^{[m]} \rightarrow \mathbb{N}_{\ge 0}$ over the set of goods. We represent a fair division problem instance by $\Ins$.

\noindent
\textbf{Allocations.} An allocation, $\mathcal{A} \coloneqq (A_1, \ldots, A_n)$ is a partition of all the goods among the $n$ agents, i.e. for all $i, j \in [n]$ with $i \neq j$, $A_i \cap A_j = \emptyset$ and $\cup_{i \in [n]}A_i = [m]$. We denote the set of all allocations by $\Pi_{[n]}([m])$. We also define a \textit{partial allocation,} denoted by $\mathcal{P} = (P_1, \ldots, P_n),$ as a partition of any subset of goods, that is, where $P_i \cap P_j = \emptyset$ for all $i \neq j$ and $\cup_{i \in [n]} P_i \subseteq [m]$. Finally, we use the notion of \textit{non-wasteful allocations} also defined in \cite{barman2020existence}. These are allocations where the marginal utility of all the goods in every bundle is non-zero, that is, for such an allocation $\mathcal{A},$ $v(g \vert A_i\backslash \{g\}) > 0$ for all $i \in [n]$ and any $g \in A_i$. 

\noindent

We now define the fairness notions we use in this work. 
\subsection{Fairness Notions}
\paragraph{Any Price Share (\APS)} Let $\Price$ denote the simplex of price vectors over the set of goods $[m]$, formally, $\Price=\{(p_1,\dots,p_m)\ge 0\ |\ \sum_i p_i=1\}$. Informally, Any Price Share is the value that an agent can guarantee themselves at {\em any} price with the budget of $\frac{1}{n}$. 
Formally, for an instance $\Ins,$ the \APS value of agent $i$ is defined as,
\begin{equation}\label{def:aps-price}
    \APS_i^n([m]) \coloneqq \min_{p \in \Price} \max_{S \subseteq [m], p(S) \leq \frac{1}{n}} v_i(S)
\end{equation}
where $p(S)$ is the sum of prices of goods in $S$. We will refer $\APS_i^{[n]}([m])$ by $\APS_i$ when the qualifiers $n$ and $m$ are clear.

An alternate definition without using prices is as follows. 
\begin{definition}[Any Price Share]\label{def:aps-sets}
 The \APS value of an agent $i$ for an instance $\Ins$ is the solution of the following program.
\begin{align*}
& \APS_i = \max z \\ 
    \text{subject to: }& \sum_{T \subseteq [m]} \lambda_{T} = 1 \\
    & \lambda_T = 0 &&\forall T \text{ such that } v_i(T) < z \\
    & \sum_{T \subseteq [m]: j \in T} \lambda_T \leq \frac{1}{n} &&\forall j \in [m] \\
    & \lambda_T \geq 0 &&\forall {T \subseteq [m]}
\end{align*}
\end{definition}
Essentially, an agent must decide the maximum value $z$ that satisfies the following. They associate non-zero weights $\lambda_T$ to all the sets $T\subseteq [m]$ such that any set with a value less than $z$ has weight zero, the sum of the weights on all the sets is $1,$ and the total weight on any good, defined as the sum of weights of the sets containing the good, is at most $\frac{1}{n}$. This maximum value of $z$ is their \APS value.

Both of these definitions and their equivalence is stated in \cite{BabaioffEF21}. 

\paragraph{Maximin Share (\MMS)}The Maximin share (\MMS) value of an agent $i$ for an instance $\Ins$ is defined as the minimum value they can guarantee while partitioning all the goods into $n$ bundles, assuming they pick the worst bundle in any allocation. Formally,

\[\MMS^n_i([m]) = \displaystyle\max_{(A_1,\dots,A_n) \in \Pi_{[n]}([m])} \displaystyle\min_{k \in [n]} v_i(A_k).\]    

We refer to $\MMS_i^n([m])$ by $\MMS_i$ when the qualifiers $n$ and $m$ are clear. We will refer to the allocation that defines the \MMS value of any agent $i$, that is, $\argmax_{\Pi_{[n]}([m])}\min_{k \in [n]} v_i(A_k),$ as the \textit{\MMS defining allocation} of agent $i.$

\noindent
Note that both the \APS and \MMS values of an agent do not depend on the valuation functions of the other agents, and depend only on the number of agents in the fair allocation instance. 

The following relation between the \APS and \MMS values of any agent is known.
\begin{claim}\cite{BabaioffEF21}\label{clm:aps-dominates-mms}
For any monotone valuation function $v_i$ of agent $i$, we have $\APS_i \geq \MMS_i$.
\end{claim}

At times we abuse notation, and refer as 
\textit{the \APS (or \MMS) value of a function $v(\cdot)$,} which essentially is a value of an agent whose valuation function is $v(\cdot)$. 

\subsection{Valuation Functions}
\noindent{\bf Binary marginals.}
For a valuation function $v(\cdot),$ the \textit{marginal utility} of a good $g \in [m]$ over a set $S \subseteq [m]$, denoted by $v(g \vert S),$ is defined as the increase in the total value of the bundle $S\cup \{g\}$ over the set $S,$ that is, $v(g \vert S)=v(S \cup \{g\}) - v(S)$. We consider valuation functions with \textit{binary} marginals, defined as those for whom the marginal utility of any good over any set is either $0$ or $1,$ that is, $v(g\vert S)\in \{0,1\}$ for every $g\in [m]$ and $S\subseteq [m].$ 

\noindent 
\textbf{Submodular functions.} 
A function is called \textit{submodular} if it satisfies the property of \textit{diminishing marginal returns}, which specifies that the marginal utility of any good $g$ over any subset of goods $S$ must not be larger than its marginal utility over any subset of $S.$ Formally, a function $v:2^{[m]}\rightarrow \mathbb{R}_{\ge 0}$ is called submodular if and only if,
$$v(g\vert S)\le v(g\vert S'), \ \forall g\in [m],\ S'\subseteq S\subseteq [m].$$

\noindent
\textbf{\XOS functions.} A function $v:2^{[m]}\rightarrow \mathbb{R}_{\ge 0}$ is called \textit{additive}, if the value of a set of goods is equal to the sum of values of the goods in the set, that is, $v(S) = \sum_{g\in S}v(\{g\}).$ A function $v:2^{[m]}\rightarrow \mathbb{R}_{\ge 0}$ is said to be \XOS, or \textit{fractionally subadditive,} if and only if there exists a family of additive set functions $\mathcal{F},$ such that the value of each subset $S\subseteq [m]$ is the maximum function value of $S$ from the functions in $\mathcal{F},$ that is, $v(S)=\max_{f\in\mathcal{F}}f(S).$ Note that the cardinality of the family $\mathcal{F}$ can be exponentially high in $m.$ 

We focus on submodular and \XOS functions with binary marginals. In case of the submodular functions, these are equivalent to 
what are called \textit{matroid rank functions (\MRFs)} (See Section \ref{sec:prelim-matroid} for a matroid based definition of these functions). 

For the submodular functions, we also study the case of \textit{ternary marginals,} where all the marginal values $v(g\vert S) \in \{0,1/2,1\}$ 
for all goods $g\in [m]$ and $S\subseteq [m].$

\subsection{Matroid Preliminaries}\label{sec:prelim-matroid}
\paragraph{Matroid.}\label{def:matroid} A matroid, denoted by $\M,$ is a tuple $(E, \I)$ where $E$ is a set of elements, called the \textit{ground set}, and $\I \subseteq 2^E$ is a collection of subsets of $E$ called the \textit{independent sets} of the matroid, that satisfies the following properties.
\begin{enumerate}
    \item If $S \in \I$ then $\overline{S} \in \I$ for all $\overline{S} \subseteq S$.
    \item If $I, J \in \I$ and $|I| > |J|$ then there exists $i \in \{I \setminus J\}$ such that $J \cup \{i\} \in \I$.
\end{enumerate}

\paragraph{Bases of a matroid.}Any independent set of the largest cardinality, that is any set $B\in \argmax_{I\in \I}|I|,$ is called a \textit{base} of the corresponding matroid. 

\paragraph{Rank of a matroid.} Every matroid $\M = (E, \I)$ has what is called a \textit{rank function} associated with it, that maps any subset $S$ of the ground set $E$ to a non-negative integer, equal to the size of the largest independent set that is a subset of $S$. We denote this function by $r_{\M}: 2^{E} \rightarrow \mathbb{Z}_{\geq 0}.$ Formally,
\begin{equation}
    r_{\M}(S) := \max_{I \subseteq S, I \in \I} |I|
\end{equation}
The \textit{rank} of a matroid $\M$ is the value $r_{\M}(E).$

\paragraph{Matroid rank functions and submodularity.} It is well known that the rank function of a matroid $\M=(E,\I)$ is equivalent to a submodular function on a set of $|E|$ items with binary marginals \cite{schrijver2003combinatorial}. That is, any submodular function on a set of $m$ items with binary marginals corresponds to a matroid with a ground set of $m$ elements, one corresponding to each item, referred as the \textit{underlying matroid} hence forth. The submodular function's value for any subset of the items is equal to the value of the rank function of the underlying matroid for the set of elements corresponding to the items.

\paragraph{Matroid Union.} The union function applied to a collection of matroids generates a matroid known as the \textit{union matroid}. Let $\M^{\cup k} = (E^{\cup k}, \I^{\cup k})$ denote the union of a collection of matroids $\M_1 = (E_1, \I_1), \ldots, \M_k = (E_k, \I_k).$ $\M^{\cup k}$ is defined as,
\begin{equation*}
    E^{\cup k} = \cup_{i \in [k]} E_i,\ \ \ \  \I^{\cup k} = \{\cup_{i \in [k]} I_i \vert I_i \in \I_i\}.
\end{equation*}
Essentially, for the union matroid, the ground set is the union of the ground sets of the matroids in the collection, and the independent sets are all possible sets formed by taking the union of one independent set from each underlying matroid. 

\paragraph{Rank function of a Union matroid.} The rank function of a union matroid $\M^{\cup k}$, denoted by $r_{\M^{\cup k}}(\cdot)$ or simply $r_{\M}(\cdot)$ when the underlying matroids are clear, has the following well known formula.
\begin{equation}
    r_{\M}(S) = \min_{T \subseteq S} \left[ |S \setminus T| + \sum_{i \in [k]} r_{\M_i}(T \cap E_i)  \right]
\end{equation}
Here for each $i\in [k],$ $r_{\M_i}(\cdot)$ is the rank function of the underlying matroid $\M_i$.

\paragraph{Union of copies of a matroid.} Given a matroid $\M=(E,\I),$ let $\M^n$ denote the union of $n$ copies of $\M.$ Let $r(\cdot), r(\M)$ and $r^n(\M^n)$ respectively denote the rank function of $\M,$ the rank of $\M,$ and the rank of $\M^n.$ The following properties relating these quantities are well known \cite{schrijver2003combinatorial}.
\begin{lemma}\label{lem:matroid-union-rank-n}
$r^n(\M^n) = n \cdot r(\M)$ if and only if for all subsets $T \subseteq E,$
\begin{equation}\label{eq:matroid-union-rank-n}
    |E \setminus T| \geq n \cdot [r(E) - r(T)].
\end{equation}
\end{lemma}

\begin{lemma}\label{lem:matroid-disjoint-bases}
If $r^n(\M^n) = n \cdot r(\M)$, then $\M$ has at least $n$ disjoint bases.
\end{lemma}

\section{Submodular Valuations with Binary Marginals (Matroid Rank Functions)}\label{sec:mrf}
In this section, we will prove Theorem \ref{thm:matroid-rank}. As a corollary, we get the computational result of Theorem \ref{thm:mrfs-computation}.

\begin{theorem}\label{thm:matroid-rank}
If the valuation function $v_i$ of agent $i$ is a submodular function with binary marginals, a.k.a. matroid rank function, then their $\APS$ and $\MMS$ values are equal, {\em i.e.,} $\MMS_i=\APS_i.$
\end{theorem}

\noindent 
\textbf{Proof Idea.} Recall that such a $v_i$ is a {\em matroid rank function}. To prove this theorem, we consider the underlying matroid of the valuation function $v_i.$ We use the set based definition of $\APS_i,$ Definition \ref{def:aps-sets}, to show that equation \eqref{eq:matroid-union-rank-n} is true for $\M$. This is the most crucial and technically involved step in the proof. With this equation, it then follows from Lemma \ref{lem:matroid-union-rank-n} that the ranks of $\M$ and the union matroid of $n$ copies of $\M,$ say respectively $r(\M)$ and $r^n(\M^n),$ satisfy $r^n(\M^n)=n \cdot  r(\M).$ Lemma \ref{lem:matroid-disjoint-bases} then implies that $\M$ has at least $n$ disjoint bases. These bases translate to bundles of the goods in the fair allocation instance such that the value of $i$ for each base is at least equal to $r(\M).$ From the definition of $\MMS,$ $\MMS_i$ is thus at least $r(\M).$ Finally we show $r(\M)$ is equal to $\APS_i.$ Combining with Claim \ref{clm:aps-dominates-mms} proves the Theorem.

In the remaining section we discuss the proof in detail, using the above notations. . 

A key notion towards establishing equation \eqref{eq:matroid-union-rank-n} is \textit{capping the valuation function} of $i.$ Using $v_i(\cdot)$, we define a new function $\widehat{v}_i(\cdot)$ as,
\begin{equation}\label{eqn:hat-v}
    \widehat{v}_i(S) = \min \{v_i(S), \APS_i\}
\end{equation}
We first claim that capping $v_i$ maintains the matroid rank property. 
\begin{restatable}{lemma}{LemHatV}
\label{lem:hat-v-mrf}
If $v_i$ is an \MRF then $\widehat{v}_i$ as defined in Equation \ref{eqn:hat-v} is also an \MRF.
\end{restatable}
\begin{proof}
We will show that $\widehat{v}_i(\cdot)$ is a submodular function with binary marginals, hence equivalently is an \MRF. 

Consider any good $g \in [m]$ and set $S \subseteq [m]$. We have,
\begin{align*}
    \widehat{v}_i(g \vert S) &= \widehat{v}_i(g \cup S) - \widehat{v}_i(S) \\
    &= \min\{v_i(g \cup S), \APS_i\} - \min \{v_i(S), \APS_i\}.
\end{align*}
Now if $\min \{v_i(S), \APS_i\} = \APS_i,$ then by the monotonicity of $v_i(\cdot)$, $\min \{v_i(g \cup S), \APS_i\} = \APS_i,$ implying their difference is zero, and $\widehat{v}_i(g \vert S)=0$. 

Otherwise, if $v_i(S) < \APS_i,$ then $v_i(g \cup S)$ is at most $1$ more than $v_i(S),$ as $v_i(\cdot)$ has binary marginals. Therefore, $\min\{v_i(g \cup S), \APS_i\}$ is also at most $1$ more than $v_i(S)=\min \{v_i(S), \APS_i\},$ their difference is at most one, hence $\widehat{v}_i(g \vert S)\le 1.$ This shows that $\widehat{v}_i(\cdot)$ has binary marginals. 

It is left to show that $\widehat{v}_i(\cdot)$ is submodular. We consider any set $S'$ that is a superset of $S,$ and show $\widehat{v}_i(g\vert S)\ge \widehat{v}_i(g\vert S').$ 

If $g\in S,$ then both of these values are zero. Also when $\widehat{v}_i(g\vert S)=1,$ then as $\widehat{v}_i(\cdot)$ has binary marginals, the inequality follows easily. Finally, suppose $g\notin S,$ and $\widehat{v}_i(g\vert S)=0.$ Similarly as $\widehat{v}_i(g\vert S),$ we have,
\begin{align*}
    \widehat{v}_i(g|S')
    = \min\{v_i(g \cup S'), \APS_i\} - \min\{v_i(S'), \APS_i\}
\end{align*}
Here if $v_i(S)\ge \APS_i,$ then from the monotonicity of $v_i,$ all the terms  $v_i(S'\cup \{g\}), v_i(S'), v_i(S\cup \{g\})\ge \APS_i,$ and both the marginal utilities $\widehat{v}_i(g\vert S)$ and $\widehat{v}_i(g\vert S')$ are zero. 

Otherwise, when $v_i(S)<\APS_i,$ then as $\widehat{v}_i(g\vert S)=0,$ we have $\min \{v_i(g\cup S),\APS_i\}-v_i(S)=0.$ Again as  $v_i(S)<\APS_i,$ $\min \{v_i(g\cup S),\APS_i\}=v_i(g\cup S).$ Therefore, $\widehat{v}_i(g\vert S)=v_i(g\vert S),$ and $v_i(g\vert S)$ is also zero. As $v_i(\cdot)$ is submodular, $v_i(g\vert S)\ge v_i(g\vert S').$ As $v_i(\cdot)$ has binary marginals, $v_i(g\vert S')=0.$ We have,
\begin{align*}
    v_i(g\vert S')&=v_i(g\cup S')-v_i(S')\\ 
    &\ge \min\{v_i(g\cup S'),\APS_i\}-\min\{v_i(S'),\APS_i\}\\
    &=\widehat{v}_i(g\vert S').
\end{align*}
Thus $\widehat{v}_i(g\vert S')=0,$ and therefore $\widehat{v}_i(g\vert S')=\widehat{v}_i(g\vert S).$
\end{proof}

Next, we relate the $\APS$ values of $i$ under $\widehat{v}_i$ to their corresponding values under $v_i.$ Let $\APS$ value of $i$ under $\widehat{v}_i$ be $\overline{\APS}_i$.
\begin{restatable}{claim}{ClmAPSVHatV}
\label{clm:aps-hat-v}
$\overline{\APS}_i = \APS_i.$
\end{restatable}
\begin{proof}
First note that since $\widehat{v}_i$ is capped at $\APS_i$, the value of any set, in particular the best bundle she can afford at any price vector, cannot be greater than $\APS_i$. Thus, from the price based definition of $\APS$ shown in equation \eqref{def:aps-price}, $\overline{\APS}_i \leq \APS_i$.

To see the other direction, consider Definition \ref{def:aps-sets} of $\APS$. Since the $\APS$ value of the function $v_i$ is $\APS_i$, there exist some $k$ sets $\Set=\{S_1, \ldots, S_k\}$ each of value at least $\APS_i$ under $v_i,$ and corresponding weights $\Lambda=\{\lambda_1, \ldots, \lambda_k\}$ that satisfy the constraints in Definition \ref{def:aps-sets}. Now, even under $\widehat{v}_i$, these sets have value at least $\APS_i$, in fact, exactly $\APS_i$. We show that the sets $\Set$ and their weights $\Lambda$ form a feasible solution to the program of Definition \ref{def:aps-sets} for $z=\APS_i,$ even under $\widehat{v}_i.$ 

The only constraints that depend on the valuation function are, $\lambda_T=0$ for all $T$ where $\widehat{v}_i(T)\le z.$ As we have fixed $z=\APS_i,$ and as $\widehat{v}_i(T)\le v_i(T)$ for any set $T,$ these constraints hold. The remaining constraints hold trivially. Therefore, $\overline{\APS}_i \geq z= \APS_i.$
\end{proof}
Analogously, we relate the $\MMS$ values of $i$ under $\widehat{v}_i$ to their corresponding values under $v_i.$ Let $\MMS$ value of $i$ under $\widehat{v}_i$ be $\overline{\MMS}_i$.
\begin{restatable}{claim}{ClmVHatVMMS}
\label{clm:mmm-at least-mms-hat}
$\MMS_i \geq \overline{\MMS}_i.$
\end{restatable}
\begin{proof}
Consider any $\MMS$ defining allocation $\PA$ under $\widehat{v}_i$. The minimum valued bundle in this allocation has value $\overline{\MMS}_i$ according to $\widehat{v}_i$. By definition of $\widehat{v}_i$, $v_i(S) \geq \widehat{v}_i(S)$ for every set $S \subseteq [m]$. Thus, the same allocation has a value at least $\overline{\MMS}_i$ even under $v_i$. From the definition of $\MMS,$ $\MMS_i$ is at least equal to the minimum bundle's value under $v_i$ from the allocation $\PA,$ hence at least $\overline{\MMS}_i.$
\end{proof}

Let $\M_{\widehat{v}}$ denote the underlying matroid of the function $\widehat{v}_i,$ and let $\M_{\widehat{v}}^{n}$ be the union matroid of $n$ copies of $\M_{\widehat{v}}$. 

\begin{lemma}\label{lem:rank-equal-to-aps}
$r({\M_{\widehat{v}}}) = \APS_i.$
\end{lemma}
\begin{proof}
As $\widehat{v}_i$ caps the valuations at $\APS_i,$ no independent set can have size more than $\APS_i,$ hence the rank of $\M_{\widehat{v}}$ is at most $\APS_i$. At the same time, from Claim \ref{clm:aps-hat-v}, $\overline{\APS}_i = \APS_i.$ The set based Definition \ref{def:aps-sets} of $\APS$ then implies that there exist some sets of value at least $\APS_i$ under $\widehat{v}_i$. This implies that the rank $r_{\M_{\widehat{v}}}([m]) \geq \APS_i.$ Together, we get $r_{\M_{\widehat{v}}}([m]) = \APS_i$.
\end{proof}

\noindent
We now prove the key lemma towards proving Theorem \ref{thm:matroid-rank}.

\begin{lemma}\label{lem:eq-matroid-rank-in-terms-of-aps}
For any subset $T \subseteq [m]$,
\begin{equation}\label{eqn:elements-outside-T-1}
    |[m] \setminus T| \geq n \cdot [\APS_i - r_{\M_{\widehat{v}}}(T)]
\end{equation}
\end{lemma}
\begin{proof}
From Claim \ref{clm:aps-hat-v}, the $\APS$ value under $\widehat{v}_i(\cdot)$ is $\APS_i$. Definition \ref{def:aps-sets} of $\APS$ shows that there exists an optimal feasible solution to the program. Let the sets and their associated weights in the solution be $\Set=\{S_1, \ldots, S_k\}$ and  $\Lambda=\{\lambda_1, \ldots, \lambda_k\}$ respectively. We have $\widehat{v}_i(S_j) \geq \APS_i$ for all $j\in [k],$ and the total weight on any particular good $g$, that is, $\sum_{j:g\in S_j}\lambda_j,$ is at most $\frac{1}{n}$. As $r_{\M_{\widehat{v}}}(S_j) = \APS_i,$ as $\widehat{v}_i(S_j) = \APS_i$ for all $j \in [k]$. That is, the sets $S_j$ are bases of the matroid $\M_{\widehat{v}}.$

Consider any set $T \subseteq [m]$. Let the rank of $T$ be $t_r$. As $r(\M_{\widehat{v}})=\APS_i$ from Lemma \ref{lem:rank-equal-to-aps}, $t_r \leq \APS_i$.  Therefore, using Property $2$ of Matroid definition \ref{def:matroid}, we can move $\APS_i - t_r$ elements from each $S_j$ to $T$. Let $S'_j$ be any set of $\APS_i - t_r$ elements that can be added to $T$, with a marginal utility of one for each element. 

The total weight of all the distinct elements in the sets $S'_j$ can be expressed in two ways as,
$$\sum_{g\in \cup_j S'_j}\sum_{j:g\in S'_j}w(g,S'_j)=\sum_{j\in [k]}\sum_{g\in S'_j}w(S'_j),$$ 
where $w(S'_j)$ is the weight of the set $S'_j,$ and $w(g,S'_j)$ is the weight on good $g$ due to it belonging in $S'_j,$ in the solution $(\Set,\Lambda)$ to the program defining $\APS$. 

As the total weight on each good is at most $1/n,$ the left expression can be evaluated as,
$$\sum_{g\in \cup_j S'_j}\sum_{j:g\in S'_j}w(g,S'_j)\le \sum_{g\in \cup_j S'_j}\frac{1}{n}=|\cup_{g\in [k]}S'_j|\cdot \frac{1}{n}.$$

As $w(S'_j)$ is equal to the weight of the set $S_j,$ the right expression can be evaluated as,
$$\sum_{j\in [k]}\sum_{g\in S'_j}w(S'_j) =\sum_{j\in [k]}|S'_j|\lambda_j.$$

Equating the two, we get,
\begin{align*}
    \sum_{j \in [r]} \lambda_j |S'_j| \leq \frac{1}{n} \cdot |\cup_{j \in [r]} S'_j|.
\end{align*}
As $|S'_j| = \APS_i - t_r,$ and $\sum_{j \in [r]} \lambda_j = 1$,
\begin{align}\label{eqn:elements-outside-T-2}
    |\cup_{j \in [r]} S'_j| \geq n \cdot (\APS_i - t_r)
\end{align}
Finally, as all the sets $S'_j$ $j\in [k],$ add elements to $T$ that are not already present in $T$, we have that $\cup_{j \in [k]}S'_j \subseteq [m] \setminus T$. Thus, $|\cup_{j \in [k]} S'_j| \geq |[m] \setminus T|$. Substituting in equation \eqref{eqn:elements-outside-T-2},
\begin{align*}
    |[m] \setminus T| \geq n \cdot (\APS_i - t_r) = n \cdot [\APS_i-r_{\M_{\widehat{v}}}(T)]. &&\qedhere 
\end{align*}
\end{proof}

\noindent
Theorem \ref{thm:matroid-rank} follows as a combination of all the lemmas.

\begin{proof}[Proof of Theorem \ref{thm:matroid-rank}] First, by substituting Lemma \ref{lem:rank-equal-to-aps} in equation \eqref{lem:eq-matroid-rank-in-terms-of-aps}, and combining with Lemma \ref{lem:matroid-union-rank-n} we immediately get the following relation.
\begin{equation}\label{lem:union-matroid-rank-final}
r(\M_{\widehat{v}}^{n}) = n \cdot r(\M_{\widehat{v}}).
\end{equation}

Combining this with Lemma \ref{lem:matroid-disjoint-bases}, we get that $\M_{\widehat{v}}$ has at least $n$ disjoint bases. This means we can create a partition of $[m]$ where each part has value $r(\M_{\widehat{v}_i}),$ which is equal to $\APS_i$ from Lemma \ref{lem:rank-equal-to-aps}. Thus, $\overline{\MMS}_i \geq \APS_i$. Along with Lemma \ref{clm:mmm-at least-mms-hat} we then have, $\MMS_i \geq \APS_i$. Finally, we know from Claim \ref{clm:aps-dominates-mms} $\APS_i \geq \MMS_i.$ Therefore $\APS_i = \MMS_i$.
\end{proof}
Finally, we prove as a corollary of Theorem \ref{thm:matroid-rank}, the following computational result.

\begin{theorem}\label{thm:mrfs-computation}
Given a fair allocation instance $\Ins$ where every agent's valuation function is an \MRF, there is an efficient (polynomial time) algorithm to compute an \MSW allocation where every agent $i$ receives a bundle of value at least $\APS_i.$ 
\end{theorem}
\begin{proof}
We know from \cite{BarmanV21} that an \MSW allocation that gives every agent a bundle of value at least $\MMS_i$ exists and can be computed efficiently. Theorem \ref{thm:matroid-rank} implies the same allocation guarantees an $\APS$ or higher valued bundle to each agent.
\end{proof}

\section{\XOS Valuations with Binary Marginals}\label{sec:xos}
In this section, we consider $\APS$ approximations when agents have $\XOS$ valuations with binary marginals. Theorems \ref{thm:xos-main}, \ref{thm:xos-aps-lower-bound} and \ref{thm:xos-algo} are the main results of this section.

\begin{theorem}\label{thm:xos-main} If the valuation function of an agent $i$ in an instance $\Ins$ is an $\XOS$ function with binary marginals, then their $\APS$ and $\MMS$ values satisfy, $\APS_i \leq 2 \cdot \MMS_i + 1.$
\end{theorem}

\noindent
\textbf{Proof Idea.} Recall the notions of partial and wasteful allocations from Section \ref{sec:prelim}. The crucial step in the proof is Algorithm $\ref{alg:alloc-mms-price},$ which takes as input a fair allocation instance $\Ins,$ and yields a non-wasteful, partial allocation where each of the allocated bundles has value at most $\MMS_i + 1,$ and the set of unallocated goods has a value of at most $\MMS_i$. Using this allocation, we fix prices on the goods such that for agent $i$ the highest value of any affordable bundle of goods, that is one with total price at most $1/n,$ is at most $2\MMS_i+1.$ The price based definition shown in equation \eqref{def:aps-price} of $\APS,$ then implies the theorem. 

Let us now discuss the details of the proof. Hence forth, we call an allocation \textit{balanced} for an agent, if the difference in the values of the smallest and largest bundles according to the agent's valuation function is at most $1.$ 

\begin{algorithm}[tbh!]
\caption{Non-wasteful balanced \MMS allocation}\label{alg:alloc-mms-price}
\DontPrintSemicolon
  \SetKwFunction{Define}{Define}
  \SetKwInOut{Input}{Input}\SetKwInOut{Output}{Output}
  \Input{$([n], [m], v_i(\cdot))$ where $v_i$ is a binary $\XOS$ valuation function}
  \Output{A non-wasteful, balanced \MMS allocation according to $v_i,$ where the leftover goods also have value at most $\MMS_i + 1.$}
  \BlankLine
  Initialize $\A = (A_1, A_2, \ldots, A_n)$ to be any $\MMS$-defining allocation for
  agent $i$\;\label{step:init-mms}
  Initialize $\R \gets \emptyset$\;
  \For{$j \in [n]$}{\label{step:for-start}
  \If{$v_i(A_j) \geq \MMS_i + 1$}{
    Let $G \in \{A \subseteq A_j \vert v_i(A) = |A| = \MMS_i$\}\;\label{step:select-G}
    Set $A_j \gets G$\;
    $\R \gets \R \cup \{A_j \setminus G\}$\;
    }}\label{step:for-end}
  $\R \gets [m] \setminus \cup_{j \in [n]}A_j$\;
  \While{$v_i(\R) \geq \MMS_i + 1$}{\label{step:while-start}
    Let $A' \subseteq \R$ with $v_i(A') = |A'| = \MMS_i + 1$\; \label{step:select-R}
    Let $k \gets \arg\min_{j \in [n]} v_i(A_j)$\;
    $\R \gets \R \cup A_k$\;
    $\A_k \gets A'$\;
  }\label{step:while-end}
\Return $\A$
\end{algorithm}

\noindent
\textbf{Algorithm.} The algorithm starts by computing any $\MMS$-defining allocation $\A$ for agent $i,$ and performs two phases. First, while any bundle $A_j$'s value in $\A$ is more than $\MMS_i+1,$ it considers any subset $G$ of $A_j$ that has both size and value exactly $\MMS_i,$ leaves $G$ with agent $j$ and removes the remaining goods. All goods removed in this way are added to a bundle called $\R.$ In the second phase, while the value of $\R$ is higher than $\MMS_1+1,$ it considers any subset $A'$ of $\R$ of both size and value exactly $\MMS_i+1.$ The algorithm takes away the bundle of the smallest valued agent, adds this to $\R,$ and gives $A'$ to this agent instead. 

We use two results that will be useful in establishing that this algorithm converges in the special kind of allocation desired. Claim \ref{clm:xos-prop} is a technical property of $\XOS$ valuations, followed by Lemma \ref{lem:xos-partial-alloc-bound} which shows a key property of Algorithm \ref{alg:alloc-mms-price}.
\begin{restatable}{claim}{ClaimNonWastefulXOS}
\label{clm:xos-prop}
Given an $\XOS$ valuation function  $v: 2^{[m]} \rightarrow \mathbb{R}_{\ge 0}$ with binary marginals, and any set $S \subseteq [m]$, we can find a subset $S' \subseteq S$ such that $v(S') = v(S) = |S'|$.
\end{restatable}
\begin{proof}
Since $v$ is an $\XOS$ function, there is a family $\F$ of additive functions such that for all $\Set \subseteq [m]$, there is an additive function $f \in \F$ with $v(\Set) = f(\Set)$. Given the set $S$ consider any such function $f$. As $f$ is additive, $f(S)=\sum_{g\in S}f(g).$ Further, as $v$ has binary marginals, $f(g)\in \{0,1\}$ for all $g.$ We define the set $S' \coloneqq \{g \in S \vert f(g) = 1\}$. Thus we have, $v(S) = f(S) = f(S') = |S'|$. 

Finally, as $f(S')=|S'|,$ $v(S')\ge |S'|,$ but as $v$ has binary marginals, $v(S') \leq |S'|$, thus $v(S')=|S'|.$
\end{proof}
\begin{remark}
    We note here that while Claim \ref{clm:xos-prop} seems obvious, it is not true for binary subadditive valuations which are the immediate generalisation of binary $\XOS$ valuations. To see this, consider a function on $m = 3$ goods where the entire set of goods is valued at $2$ and any strict subset of the three goods has a value $1$. One can verify that this function is subadditive and that Claim \ref{clm:xos-prop} does not hold for this function.
\end{remark}
We next prove the following lemma.
\begin{restatable}{lemma}{LemmaAlgorithmProp}\label{lem:xos-partial-alloc-bound}
Algorithm $\ref{alg:alloc-mms-price}$ terminates, and the output allocation $\A$ is a non-wasteful, balanced, \MMS allocation according to $v_i$. Furthermore, $v_i(\R) \leq \MMS_i$.
\end{restatable}
\begin{proof}
Consider the first For loop (Steps \ref{step:for-start} to \ref{step:for-end}). This loop accesses every bundle $A_j$ at most once, and if its value is more than $\MMS_i + 1$, finds a subset $G\subseteq A_j$ of value $\MMS_i$. From Claim \ref{clm:xos-prop}, such a set $G$ always exists. Since the algorithm is not required to execute in polynomial time, let us assume that the algorithm finds the set by enumeration. Thus, at the end of this loop, we have $v_i(A_j) = \MMS_i$ for each $j \in [n]$.

Let us now look at the While loop (Steps \ref{step:while-start} to \ref{step:while-end}). If this loop has $n$ or more iterations, then after $n$ iterations it replaces more than $n$ bundles, and finds a new (partial) allocation where every bundle has value equal to $\MMS_i+1.$ From the definition of \MMS, this means $\MMS_i>\MMS_i+1,$ a contradiction. Therefore, this loop is executed at most $n-1$ times and in the end, we have $v_i(A_j) \leq \MMS_i + 1$ for all $j$. After the loop terminates, the condition for staying in the loop is falsified, hence $v_i(\R) \leq \MMS_i$.

Therefore, in at most $O(n)$ iterations of the For and While loops, the algorithm terminates, and yields an allocation where $\MMS_i \leq v_i(A_j) \leq \MMS_i + 1$ for all $j \in [n],$ and $v_i(\R) \leq \MMS_i,$ that is, a non-wasteful, balanced and \MMS allocation.
\end{proof}

\begin{proof}[Proof of Theorem \ref{thm:xos-main}]
Consider the partial allocation $\A$ and the set of remaining goods $\R$ obtained at the end of Algorithm \ref{alg:alloc-mms-price}. 
We define a price vector $\vecp=(p_j)_{j\in[m]}$ for the goods based on $\A$ as follows. 
\begin{align*}
    & p_j = \left\{ \begin{array}{ll}
         \frac{1}{n\cdot |A_k|},& \forall j\in A_k, \forall A_k\in \A \\
             0,&\quad \forall j\in \R.
    \end{array}\right.
\end{align*}
Let us see the maximum value that an agent with a budget $1/n$ can afford with this price vector. First, they can get all of $\R$ for free. From Lemma \ref{lem:xos-partial-alloc-bound}, $\R$ has value at most $\MMS_i$ for $i.$ Further, each bundle $A_k$ has value under $v_i$ at most $\MMS_i + 1.$
Therefore, each good in $\A$ has price at least $1/(n\cdot (\MMS_i+1)).$ Even if agent $i$ picks all the lowest priced goods and gets a marginal increment of one for each of them, at a budget of $\frac{1}{n}$, they can receive a value of at most $\MMS_i + 1$ from $\A$. By subadditivity of $v_i$, their total value from $\A$ and $\R$ together is at most $2 \cdot \MMS_i + 1$ value. From the price based definition of $\APS,$ $\APS_i\leq 2\cdot \MMS_i+1.$
\end{proof}

Next, we prove that Theorem \ref{thm:xos-main} is almost the tightest relation between $\MMS$ and $\APS$ values for this setting.

\begin{restatable}{theorem}{ThmXOSLowerBound}
\label{thm:xos-aps-lower-bound}
There exists a fair allocation instance with three agents and six goods, where all the agents have an identical \XOS valuation function with binary marginals, and their $\APS$ and $\MMS$ values satisfy $\APS\geq 2\cdot \MMS.$ 
\end{restatable}
\begin{proof}
The instance $\Ins$ is as follows. There are three agents and six goods, i.e. $n = 3$ and $m = 6$. Let the goods be denoted by $g_i,\ i\in[6].$ The identical \XOS valuation function of all the agents has two additive functions in the family $\F$, say $f_1$ and $f_2$. 

The first three goods have value $1$ under $f_1$ and the remaining have value $0,$ that is, $f_1(g_i)=1,$ for $i\in [3],$ and $f_1(g_i)=0$ for $i\in [6]\setminus [3].$ Under $f_2$, the opposite is the case, i.e., $f_2(g_i)=0,$ for $i\in [3],$ and $f_2(g_i)=1$ for $i\in [6]\setminus [3].$

As the agents are identical, they have the same $\APS$ and $\MMS$ values. Now, as there are three agents, if $\MMS$ was more than one, each agent must get at least two goods. Note that under $v$, an agent can get a value of two if they receive two goods from $\{g_1, g_2, g_3\}$ or two goods from $\{g_4,g_5,g_6\}$. But to create three bundles of size two each, at least one bundle would have one good from each set $\{g_1,g_2,g_3\}$ and $\{g_4,g_5,g_6\}$. This bundle however would have a value of $1$ under $v$. Thus, the $\MMS$ cannot be more than $1$. 

On the other hand, consider the set of $6$ sets $\{g_i, g_{i+1}\}$ for $i \in [3]$ and $i \in \{4,5,6\}$. Assign a weight of $1/6$ to each of these sets. The total weight assigned is $1.$ Also, each good belongs in exactly two sets, hence the total weight on any single good is $1/3$. Each set has value $2$. Therefore, this is a feasible solution to the Linear program in the $\APS$ definition \ref{def:aps-sets} for $z=2$. Thus, $\APS \geq 2$. Therefore, in this instance, $\APS \geq 2 \cdot \MMS$. 
\end{proof}
Theorems \ref{thm:xos-main} and \ref{thm:xos-aps-lower-bound}, along with the known computational results for $\MMS,$ yield the following results for \APS.

\begin{theorem}\label{thm:xos-algo}
Given a fair allocation instance $\Ins$ where every agent has a binary \XOS valuation function, \begin{enumerate}
    \item A $0.1222$-\APS allocation, meaning one that gives every agent a bundle of value at least $0.1222$ times their $\APS,$ can be computed in polynomial time. 
    \item Even when agents have identical valuations, no better than $0.5$-\APS allocation may exist. 
\end{enumerate} 
\end{theorem}

To prove this theorem, we first prove the following lemmas to separate the agents with $\APS=0$ and $\MMS=0.$
\begin{restatable}{lemma}{LemMMSZeroCheck}\label{lem:check-mms-0}
Given an instance $\Ins$ with \XOS binary marginal valuations, one can check in polynomial time if any agent has $\MMS=0.$
\end{restatable}
\begin{proof}
To check if $\MMS_i=0$ for some agent $i$, we form a complete weighted bipartite graph where one side has $n$ vertices, and the goods correspond to vertices on the other side. The weight of each edge is the value of the agent $i$ for the good adjacent to the edge. We compute one maximum weight matching of the goods. If each of the $n$ vertices in the left part gets assigned a good of value $1,$ then one can form an allocation with the matched goods in separate bundles. The remaining goods can be allocated arbitrarily. By subadditivity, each bundle has value at least $1,$ hence $\MMS_i\geq 0.$
\end{proof}
\begin{restatable}{lemma}{LemXOSMMSZero}
\label{lem:check-aps-0}
For an agent with a binary \XOS valuation function in a fair allocation instance, $\APS=0$ if and only if $\MMS=0.$
\end{restatable}
\begin{proof}
If $\MMS_i>0,$ then from Claim \ref{clm:aps-dominates-mms}, $\APS_i>0.$ Otherwise when $\MMS_i=0,$ the allocation returned by Algorithm \ref{alg:alloc-mms-price} has the following properties according to Lemma \ref{lem:xos-partial-alloc-bound}. The value of $\R=0,$ and the value of each bundle in $\A$ is either $0$ or $1.$ Further, there are at most $(n-1)$ bundles with value $1,$ else $\MMS_i\geq 1.$ Also, each of the $1$ valued bundles have exactly $1$ good in them, and the $0$ valued bundles have no goods in them. We assign prices to the goods as follows. Assign a price of $1/(n-1)$ to each good in the highest valued $n-1$ bundles of $\A,$ assign a price of $0$ to all the remaining goods (in $\R$).

At a budget of $1/n,$ the agent cannot afford any $1$ valued bundle. As the remaining goods together have a value $0,$ $\APS=0.$
\end{proof}

\begin{proof}[Proof of Theorem \ref{thm:xos-algo}] From Lemmas \ref{lem:check-mms-0} and \ref{lem:check-aps-0}, we remove all the agents with $\APS=0$ by giving them no goods in polynomial time. For the remaining agents, we know $\APS\geq \MMS\geq 1.$ Combined with $\APS\leq 2\cdot\MMS+1$ from Theorem \ref{thm:xos-main}, we have $\APS\leq 3\cdot\MMS$ for all the remaining agents. 

\cite{li2021fair} show that there exists an efficient algorithm to compute a $0.3666$-\MMS allocation, that is, an allocation where every agent receives a bundle of value at least $0.3666$ times their $\MMS$ value. As $\APS\leq 3\cdot\MMS,$ this implies that an allocation that gives every agent a bundle of value at least $~0.1222$ times their $\APS$ can be computed in polynomial time.   

Finally, Theorem \ref{thm:xos-aps-lower-bound} shows an instance where the agents have identical valuation functions, and their $\APS$ is at least twice their $\MMS.$ By definition of $\MMS,$ no allocation can have the smallest bundle's value more than $\MMS.$ Therefore, in every allocation, at least the smallest bundle's agent receives a bundle of value at most half their $\APS,$ and a better than $0.5$-\APS allocation does not exist.
\end{proof}
\section{Submodular Valuations with Ternary Marginals}\label{sec:tern-submod}

In this section, we show the following hardness result.
\begin{theorem}\label{thm:tern-submod}
In an instance $\Ins,$ it is \classNP-hard to compute the \MMS value of an agent approximately up to a factor better than $5/6,$ when the agent has a submodular valuation function, even when all the marginal utilities are in $\{0,1/2,1\}.$
\end{theorem}

An immediate corollary, using Claim \ref{clm:aps-dominates-mms}, is that computing the $\APS$ value approximately up to a factor better than $5/6$ is also \classNP-hard for such an agent.

The proof of the theorem has three parts. First, we show a reduction from the known \classNP-complete problem \threematch to a fair allocation instance with agents with identical valuations. We then show that this valuation function in the reduced instance is submodular with all the marginal values in $\{0,1/2,1\}.$ Finally, we show the correctness of the reduction, establishing the factor of the hardness of approximation.

\noindent
\textbf{Reduction.} The $\threematch$ problem is as follows. Given are three disjoint sets $X,Y,Z,$ having $m$ elements each, and a set $\T$ of triples $(a,b,c)$, where $a\in X,$ $b\in Y$ and $c\in Z.$ The problem is to answer if there is a subset of $m$ triples in $\T$ called a $\threematch$ of $X,Y,Z,$ that \textit{cover} all of $X,Y$ and $Z,$ meaning for all $s\in X\cup Y \cup Z,$ $s\in \threematch.$ 

Given an instance of $\threematch,$ we form a fair allocation instance as follows. There are $m$ agents, and a set $\M$ of $3m$ goods, one good corresponding to each element of $X,Y$ and $Z.$ All the agents have the following identical valuation function $v$ for the goods. For any subset $\Set$ of $\M,$ $v(\Set)$ is defined as follows.
\begin{align*}
    & v(\Set) = \left\{ \begin{array}{ll}
         1,& \text{ if }|\Set|=1 \\
         2,& \text{ if }|\Set|=2 \\ 
         2.5,& \text{ if }|\Set|=3 \text{ and }\Set\notin \T \\
         3,& \text{ if }|\Set|=3 \text{ and }\Set\in \T \\
         3,& \text{ if }|\Set|\geq 4 \\
         \end{array}\right.
\end{align*}

\textbf{Function $v$ is submodular with ternary marginals.} First, let us compute the marginal utility of a good $g$ over sets of different sizes. From the definition of $v(\Set),$ one can verify that,
\begin{align*}
    & v(g\vert \Set)= \left\{ \begin{array}{ll}
        1,& \text{ if }|\Set|\leq 1 \\
        0.5,& \text{ if }|\Set|=2 \text{ and }\Set\cup \{g\}\notin \T \\
        1,& \text{ if }|\Set|=2 \text{ and }\Set\cup \{g\}\in \T \\
        0.5,& \text{ if }|\Set|= 3 \text{ and }\Set\notin \T\\
        0,& \text{ if }|\Set|= 3 \text{ and }\Set\in \T\\
        0, &\text{ if }|\Set|\geq 4.
         \end{array}\right.
\end{align*}

Therefore, $v(g\vert \Set)\in \{0,1/2,1\}$ for all $g$ and $\Set.$

Also $v(g\vert \Set)\leq v(g\vert \Set'),$ for any two sets $\Set,\Set'$ with $|\Set|\geq |\Set'|,$ hence also when $\Set'\subseteq \Set.$ This establishes submodularity.

\noindent 
\textbf{Correctness.} Finally, we prove that the \MMS value of any agent is $3$ if and only if the $\threematch$ instance has a solution, and is at most $2.5$ otherwise. 

Suppose a $\threematch$ exists. Then one can divide the $3$ goods from each triple in the solution to every agent. Each agent receives a bundle of value $3.$ As the highest value under $v$ of any set of goods is $3,$ $\MMS\leq 3.$ Therefore, in this case, $\MMS=3.$

Alternatively, suppose the $\MMS$ value of the reduced instance is $3.$ Then note that every agent must receive exactly $3$ goods. Otherwise, some agent will receive at most $2$ goods, and have value at most $2.$ A bundle of $3$ goods has value $3$ only when the corresponding elements form a triple in the $\threematch$ instance. Also, the bundles of goods in the $\MMS$-defining allocation are disjoint, hence the triples allocated to all the agents are disjoint. Therefore, the allocation consists of goods corresponding to $m$ disjoint triples, that cover all the elements, hence form a solution of $\threematch.$

Now, if there was an algorithm that computed the \MMS value within a factor better than $5/6$ for such instances with submodular functions and ternary marginals, then given the reduced instance, the algorithm would output a value higher than $2.5$ if and only if a $\threematch$ existed. This proves Theorem \ref{thm:tern-submod}.

\section{Conclusion}
We analyzed the fairness notion of \APS for indivisible goods under submodular and \XOS functions with binary marginals, a rich and expressive class of valuation functions \cite{benabbou2021finding, barman2020existence, barman2021approximating, li2018fair}. Under binary submodular valuations, we give a rather surprising result that \APS=\MMS. This is not true for even additive valuations (with non-binary marginals.) On the other hand for fractionally subadditive functions with binary marginals, we show a gap of 2 between \APS and \MMS and show that is almost tight.

It would be interesting to study the relations between \APS and \MMS for other valuation functions, and in particular for binary subadditive valuations. Subadditive functions generalizes both fractionally subadditive and submodular functions, and remains relatively less explored. We note that, most of Section \ref{sec:xos} can be extended to work for binary subadditive functions, except Claim \ref{clm:xos-prop}. A generalization of this claim for subadditive valuations will be helpful in determining the gap between \MMS and \APS under these valuations.


Finally, a study of \APS for the case when agents are heterogeneous, i.e. each agent has a weight or endowment, under valuations beyond-additive with binary marginals is the next natural question. 
\bibliographystyle{ACM-Reference-Format}
\bibliography{literature}
\end{document}